\documentclass[journal,12pt,onecolumn,draftclsnofoot,]{IEEEtran}

\ifCLASSINFOpdf
\else
\fi

\newlength\figureheight 
\newlength\figurewidth 
\usepackage{pgfplots} 
\usepackage{pgfplotstable}
\usepackage{scalefnt}
\usepackage{blkarray}

\usepackage{graphicx}
\usepackage{amsmath}
\usepackage{amssymb}
\usepackage{bbm}
\usepackage{bm}
\usepackage{framed}
\usepackage{multirow}
\usepackage{nomencl}
\usepackage{color}
\usepackage{textcomp}
\usepackage{graphics}
\usepackage{epsfig}
\usepackage{bm}
\usepackage{epstopdf}
\usepackage{courier}
\usepackage{amsthm}
\usepackage{float}
\usepackage{latexsym,amsfonts}
\usepackage{url}
\usepackage{longtable}
\usepackage[figuresright]{rotating}
\usepackage{listings}
\usepackage{algorithm}
\usepackage{algpseudocode}
\usepackage[numbers,sort]{natbib}
\usepackage{flushend}
\usepackage{footnote}
\makesavenoteenv{tabular}

\theoremstyle{definition}

\newtheorem{remark}{Remark}
\newtheorem{proposition}{Proposition}

\begin{document}

\title{Efficient Importance Sampling for the Left Tail of Positive Gaussian Quadratic Forms}

\author{Chaouki Ben Issaid, Mohamed-Slim Alouini, and Ra\'ul Tempone \\
\thanks{This work was supported by KAUST and the Alexander von Humboldt foundation.}%
\thanks{The authors are with King Abdullah University of Science and Technology (KAUST), Computer, Electrical and Mathematical Science and Engineering (CEMSE) Division, Thuwal 23955-6900, Saudi Arabia. (e-mail: {chaouki.benissaid, slim.alouini, raul.tempone}@kaust.edu.sa).}%
\thanks{Ra\'ul Tempone is also with Alexander von Humboldt Professor in Mathematics for UQ, RWTH Aachen University, 52062 Aachen, Germany.
}
}
%\date{}

\maketitle
\thispagestyle{empty}

\begin{abstract}
Estimating the left tail of quadratic forms in Gaussian random vectors is of major practical importance in many applications. In this letter, we propose an efficient importance sampling estimator that is endowed with the bounded relative error property. This property significantly reduces the number of simulation runs required by the proposed estimator compared to naive Monte Carlo (MC), especially when the probability of interest is very small. Selected simulation results are presented to illustrate the efficiency of our estimator compared to naive MC as well as some of the well-known approximations.
\end{abstract}
\begin{IEEEkeywords}
Importance sampling, left tail, positive quadratic forms, Gaussian random vectors, bounded relative error.
\end{IEEEkeywords}
\IEEEpeerreviewmaketitle
\section{Introduction}
Quadratic forms can appear in many applications including the study of the effect of inequality between errors in terms of variance and correlation in a two-way analysis of variance \cite{box1954}, when the constrained least-squares estimator is examined \cite{HSUAN1985}, but also when investigating non-coherent detection\cite{Raphaeli1996} and combining diversity \cite[Chap. 14]{Proakis1994} in communication theory.

In general, the exact distribution of a linear combination of independent non-central chi-square variates is a challenging task. In fact, various approximations have been proposed in the literature, for example, Imhof \cite{imhof1961} and Davies \cite{davies1963}. Based on the work of Imhof \cite{imhof1961}, Davies \cite{davies1963} presented a numerical approach to invert the characteristic function of a real random variable (RV) with the aim of computing its left tail. The author showed that the method accurately produces the distribution of a central chi-squared RV for various numbers of degrees of freedom. Rice \cite{Rice1980} presented a generalization of this approach, including two numerical integration methods of inverting the characteristic function. Recently, the authors in \cite{Ramirez2019} derived a new approximation to the cumulative distribution function (CDF) of positive RVs with application to the statistics of positive-definite real Gaussian quadratic forms. As stated in \cite{Bausch2013}, there is no closed analytical expression for the CDF of a linear combination of non-central chi-square RVs, and that most of the existing methods start to fail to give an accurate result when the number of terms grows considerably.

Numerical sampling techniques, namely the Monte Carlo (MC) method, offer a feasible alternative in the absence of closed-form expressions of the quantity of interest. However, for very small probabilities encountered in many applications, MC requires large sample sizes to estimate the value of the probability and as a consequence long computational times. Alternatively, we propose in this work an efficient importance sampling (IS) estimator to estimate the probability of interest, which aims to reduce the number of required simulation runs given a certain confidence interval \cite{ch2}. To the best of our knowledge, only two works proposing IS schemes for the purpose of computing tails of quadratic forms in Gaussian random vectors have been previously published \cite{Shi2016, Shi2018}. In these works, the authors were interested in the right tail and implemented IS combined with the cross-entropy method. However, the authors did not provide any efficiency analysis of their estimators. In this letter, we are rather interested in estimating the left tail of positive quadratic forms in Gaussian random vectors, and we show that our proposed IS scheme is endowed with the bounded relative error property, a well-desired property in the field of rare events simulations.\\ 
\textbf{Contributions:} Our proposed method provides two advantages over classical approaches: $(i)$ the parameter governing the convergence in our case is explicitly determined for a fixed accuracy requirement (see Eqs. \eqref{eqs1}-\eqref{eqs2}), unlike  other methods (e.g.  \cite{imhof1961} and \cite{Ramirez2019}) where the order of truncation is often determined based on trial and error, and $(ii)$ the bounded relative error property is a guarantee that our method will accurately estimate the probability of interest, even when the probability is very small, unlike other proposed methods (e.g. \cite{davies1963} and \cite{Ramirez2019}) which may fail for certain regimes or require a bigger order of truncation (more computational time) to provide a good estimate.

The rest of this paper is organized as follows. First, we briefly describe the problem setting, and provide a lower bound for the probability of interest needed for the proof of the bounded relative error property of the proposed IS scheme. In Section \ref{section2}, we prove the efficiency of our proposed estimator. After reviewing some basic notions of IS in Section \ref{section3}, we present our approach to estimating the probability of interest in Section \ref{section4}. Prior to concluding the paper, we show some simulation results to show the efficiency of our proposed approach. Throughout this paper, vectors and matrices are are denoted by bold letters.
\section{Problem Setting}\label{section2}
Let $\bm{X} = (X_1,\dots, X_N)^{T}$ be a Gaussian random vector with PDF
\begin{align}
f_{\bm{X}}(\bm{X}) = \frac{\exp\left(-\frac{1}{2} (\bm{X}-\bm{\mu})^{T} \bm{\Sigma_X}^{-1} (\bm{X}-\bm{\mu})\right)}{\sqrt{(2\pi)^{N} |\bm{\Sigma_X}|}} ,
\end{align}
where $\bm{\mu} \in \mathbb{R}^N$ is the mean vector, $\bm{\Sigma_X } \in \mathbb{R}^{N\times N}$ is the covariance matrix, assumed to be strictly positive definite, and $|\cdot|$ represents the determinant of a matrix. For a given positive definite matrix $\bm{\Sigma} \in \mathbb{R}^{N\times N}$ and a threshold $\gamma_0 > 0$, we aim to introduce an efficient IS scheme for computing the left tail of the quadratic form $\bm{X}^{T} \bm{\Sigma} \bm{X}$ as $\gamma_0 \rightarrow 0$, i.e.,
\begin{align}\label{P}
P = \mathbb{P}(\bm{X}^{T} \bm{\Sigma} \bm{X} \leq \gamma_0).
\end{align}
Before giving a lower bound on the probability $P$, we re-write its expression more conveniently. It can be shown that the probability of interest $P$ can be re-written as \cite[Ch. 3]{QPM}
\begin{align}
P = \mathbb{P}\left(S_N = \sum_{i=1}^{N}{\lambda_i (Z_i + \alpha_i)^2} \leq \gamma_0 \right),
\end{align}
where $\bm{\Lambda} = diag(\lambda_1,\dots,\lambda_N)$ is a diagonal matrix coming from the spectral decomposition of $\bm{\Sigma_X}^{\frac{1}{2}} \bm{\Sigma} \bm{\Sigma_X}^{\frac{1}{2}}$  as $\bm{Q}^T \bm{\Lambda} \bm{Q}$ with $\bm{Q}$ being an orthogonal matrix, $\bm{Z} = \bm{Q} \bm{\Sigma_X}^{-\frac{1}{2}} (\bm{X}-\bm{\mu})$, and $\bm{\alpha} = \bm{Q} \bm{\Sigma_X}^{-\frac{1}{2}} \bm{\mu}$. Note that $\{\lambda_i\}_{i=1}^{N}$ are non-negative and that $\{Z_i\}_{i=1}^{N}$ are independent Gaussian RVs with zero mean and unit variance. At a higher level of abstraction, estimating the probability $P$ amounts to determining the CDF of a linear combination of non-central chi-squared RVs with degree 1.
\begin{remark}
If the matrix $\Sigma$ is positive semi-definite, then $A$ is also positive semi-definite. Without loss of generality, we assume that the non-zero eigenvalues are $\{\lambda_i\}_{i=1}^{d}$, where $d < N$. In this case, the probability $P$ is
\begin{align}
P = \mathbb{P}\left(S_d = \sum_{i=1}^{d}{\lambda_i (Z_i + \alpha_i)^2} \leq \gamma_0 \right).
\end{align}
\end{remark}
In the rest of this paper, we assume that $d=N$, i.e., the positive definite case, but the same reasoning applies when we simply replace $N$ with $d$ in the positive semi-definite case. The following proposition gives a lower bound on $P$.
\begin{proposition}\label{prop1}
Let $\bm{X} = (X_1,\dots, X_N)^{T}$ be a Gaussian random vector with mean $\bm{\mu}$ and covariance matrix $\bm{\Sigma_X}$, and let $\bm{\Sigma} \in \mathbb{R}^{N\times N}$ be a given positive definite matrix. For a fixed threshold $\gamma_0 > 0$, we have
\begin{align}
P = \mathbb{P}(\bm{X}^{T} \bm{\Sigma} \bm{X} \leq \gamma_0) \geq \prod_{i=1}^{N}{\left[1-Q_{\frac{1}{2}}\left(\alpha_i,\sqrt{\frac{\gamma_0}{N \lambda_i}}\right)\right]},
\end{align}
where $Q_{\nu}(\cdot,\cdot)$ is the generalized Marcum-Q function \cite[Eq.(2)]{TSAY}.
%\begin{align}
%Q_{\mu}(a,b) = \frac{1}{a^{\mu-1}} \int_{b}^{\infty}{x^{\mu} \exp\left(-\frac{x^2+a^2}{2}\right) I_{\mu-1}(a x) dx}.
%\end{align}
\end{proposition}
\begin{proof}
%The proof can be found in the technical report \cite{ch4}.
In this proof, we consider the expression of $P$ 
\begin{align}
P = \mathbb{P}\left(S_N = \sum\limits_{i=1}^{N}{\lambda_i (Z_i + \alpha_i)^2} \leq \gamma_0 \right).
\end{align} 
We have
\begin{align}
\bigcap\limits_{i=1}^{N} \left\lbrace (Z_i + \alpha_i)^2 \leq \frac{\gamma_0}{N \lambda_i} \right\rbrace \subset \left\lbrace \sum_{i=1}^{N}{\lambda_i (Z_i + \alpha_i)^2} \leq \gamma_0 \right\rbrace.
\end{align}
Using the independence of $\{Z_i\}_{i=1}^{N}$, we can write
\begin{align}
P \geq \prod_{i=1}^{N}{\mathbb{P}\left((Z_i + \alpha_i)^2 \leq \frac{\gamma_0}{N \lambda_i} \right)} = \prod_{i=1}^{N}{F_{W_i}\left( \frac{\gamma_0}{N \lambda_i}\right)} ,
\end{align}
where $F_{W_i}(\cdot)$ is the CDF of the RV $W_i = (Z_i + \alpha_i)^2$. This corresponds to the CDF of a non-central Chi-squared RV with 1 degree of freedom. Therefore, we can write
\begin{align}
P \geq \prod_{i=1}^{N}{\left[1-Q_{\frac{1}{2}}\left(\alpha_i,\sqrt{\frac{\gamma_0}{N \lambda_i}}\right)\right]}.
\end{align}
\end{proof}
\section{Review of IS}\label{section3}
Let $f_{\bm{Z}}(\cdot)$ denote the PDF of $\bm{Z}=\left(Z_1, Z_2, \dots, Z_N\right)^{T}$. Then, we can write $P = \mathbb{E}[\mathbbm{1}_{(S_N \leq \gamma_0)}]$, where $\mathbbm{1}_{(\cdot)}$ is the indicator function and $\mathbb{E}[\cdot]$ is the expectation w.r.t. the probability measure under which the PDF of $\bm{Z}$ is $f_{\bm{Z}}(\cdot)$. Therefore, the naive MC estimator of $P$ is given by
\begin{align}
\hat{P}_{MC} = \frac{1}{M} \sum_{j=1}^{M}{\mathbbm{1}_{(S_N(\omega_j) \leq \gamma_0)}},
\end{align}
where $M$ is the number of MC samples, and $\{S_N(\omega_j)\}_{j=1}^{M}$ are i.i.d. realizations of the RV $S_N$. For each realization of $S_N$, the sequence $\{Z_i(\omega_j)\}_{i=1}^{N}$ is sampled independently according to the PDF $f_Z(\cdot)$. We start by re-writing $P$ as
\begin{align}
P = \mathbb{E}^*\left[\mathbbm{1}_{(S_N \leq \gamma_0)} \mathcal{L}(Z_1, \dots, Z_N) \right],
\end{align}
where $\mathcal{L}(Z_1, \dots, Z_N) = \prod_{i=1}^{N}{\frac{f_{Z_i}(Z_i)}{f_{Z_i}^*(Z_i)}}$ is the likelihood ratio and $\mathbb{E}^*[\cdot]$ is the expectation w.r.t. the probability measure under which the PDF of $\bm{Z}$ is the biased PDF $f_{\bm{Z}}^*(\cdot)$.
Thus, the IS estimator of $P$ is given by
\begin{align}\label{ISes}
\hat{P}_{IS} = \frac{1}{M^*} \sum_{j=1}^{M^*}{\mathbbm{1}_{(S_N(\omega_j) \leq \gamma_0)} \mathcal{L}(Z_1(\omega_j), \dots, Z_N(\omega_j))},
\end{align}
where the sequence $\{Z_i(\omega_j)\}_{i=1}^{N}$ is sampled according to the biased PDF $f_{\bm{Z}}^*(\cdot)$ for each realization $j=1,\dots,M^*$.

The efficiency of the proposed IS estimator compared to naive MC can be measured by many criteria. When it comes to evaluating very low probabilities, IS estimators endowed with the bounded relative error property are desirable in the field of rare events simulations. A naive MC estimator requires a number of samples $M$ that grows as $\mathcal{O}(P^{-1})$. To achieve the same accuracy, the number of simulation runs $M^{*}$ needed by an IS estimator with a bounded relative error remains bounded, independently of $P$. Mathematically speaking, we say that the IS estimator satisfies the bounded relative error criterion if the following statement holds \cite{ch3}
\begin{align}
\underset{\gamma_0 \rightarrow 0}{\limsup} \frac{\mathbb{E}^*[\mathbbm{1}_{(S_N \leq \gamma_0)} \mathcal{L}^2(Z_1,\dots,Z_N)]}{P^2} < +\infty.
\end{align}
To have a clear idea of the gain that the proposed IS estimator achieves compared to the naive MC one, we determine the number of simulation runs required by both estimators when the accuracy requirement is fixed, e.g., when the relative error of both estimators is assumed to be the same. We start by defining the relative error of both estimators as \cite{ch1}
\begin{align}
\varepsilon &= \frac{C}{P} \sqrt{\frac{P(1-P)}{M}}, \label{eqs1} \\
\varepsilon^* &= \frac{C}{P} \sqrt{\frac{\mathbb{V}^*[\mathbbm{1}_{(S_N \leq \gamma_0)} \mathcal{L}(\bm{Z})]}{M^*}}, \label{eqs2}
\end{align}
where we take $C = 1.96$, which corresponds to a $95\%$ confidence interval, and $\mathbb{V}^*$ denotes the variance w.r.t. the probability measure under which the PDF of $\bm{Z}$ is $f_{\bm{Z}}^*(\cdot)$.
\section{Proposed IS Scheme}\label{section4}
Our approach consists on shifting the mean and scaling the variance of each variate $\{Z_i\}_{i=1}^{N}$ so that the marginal biased PDF is written as
\begin{align}\label{bias}
f_{Z_i}^*(z) = \frac{1}{\sqrt{2\pi} \sigma_i} \exp\left(-\frac{1}{2} \left(\frac{z+\alpha_i}{\sigma_i}\right)^2\right).
\end{align}
While the original PDF of $Z_i$, $\forall i=1,\dots, N$, is a standard Gaussian, the biased PDF corresponds to a Gaussian with mean $-\alpha_i$ and variance $\sigma_i^2$. In our approach, we choose the parameter $\sigma_i$, hoping that the event of interest becomes no longer rare. A possible solution is to look for $\sigma_i$ in the form  $\sigma_i^2 = \theta \frac{\gamma_0}{\lambda_i}$, where $\theta$ is a positive parameter such that we have
\begin{align}
\mathbb{E}^*\left[\sum\limits_{i=1}^{N}{\lambda_i (Z_i+ \alpha_i)^2} \right] = \gamma_0.
\end{align}
Using the linearity of the expected value and the fact that, under the new probability measure, $\{Z_i + \alpha_i\}_{i=1}^{N}$ are zero mean Gaussian RVs with variance $\sigma_i^2$, we get
\begin{align}\label{choice}
\sigma_i = \sqrt{\frac{\gamma_0}{N \lambda_i}},~ i=1,\dots, N.
\end{align} 
The above value of $\sigma_i$ is clearly non-negative since the eigenvalues $\{\lambda_i\}_{i=1}^{N}$ are all non-negative. As $\gamma_0$ approaches zero, the values of $\sigma_i$ become smaller, leading to the reduction of the variance of the IS estimator. Defining the biased PDFs using the values of $\sigma_i$ obtained in (\ref{choice}), we show that our proposed IS estimator satisfies the bounded relative error property.
\begin{proposition}
Let the marginal biased PDFs be defined as in (\ref{bias}), and $\sigma_i$ as in (\ref{choice}). Then, the IS estimator (\ref{ISes}) of the probability $P$, given by (\ref{P}), satisfies
\begin{align}
\underset{\gamma_0 \rightarrow 0}{\limsup} \frac{\mathbb{E}^*[\mathbbm{1}_{(S_N \leq \gamma_0)} \mathcal{L}^2(\bm{Z})]}{P^2} \leq \prod_{i=1}^{N}{\frac{\pi e}{\alpha_i^{2}}} < +\infty.
\end{align}
\end{proposition}
\begin{proof}
%The proof can be found in the technical report \cite{ch4}.
The likelihood ratio can be upper bounded by
\begin{align}
\nonumber &\mathcal{L}(\bm{Z}) = \prod_{i=1}^{N}{\frac{f_{Z_i}(Z_i)}{f_{Z_i}^*(Z_i)}} \\
\nonumber &= \left(\prod_{i=1}^{N}{\sigma_i}\right) \exp\left(\frac{1}{2} \sum_{i=1}^{N}{\left(\frac{Z_i+\alpha_i}{\sigma_i}\right)^2}-\frac{1}{2} \sum_{i=1}^{N}{Z_i^2}\right)\\
&\leq \left(\prod_{i=1}^{N}{\sigma_i}\right) \exp\left(\frac{1}{2} \sum_{i=1}^{N}{\left(\frac{Z_i+\alpha_i}{\sigma_i}\right)^2}\right).
\end{align}
With the choice of $\sigma_i$ given in (\ref{choice}), we get
\begin{align}
\mathcal{L}(\bm{Z}) \leq \left(\frac{\gamma_0}{N}\right)^{\frac{N}{2}} \left(\prod_{i=1}^{N}{\frac{1}{\sqrt{\lambda_i}}}\right) \exp\left(\frac{N}{2 \gamma_0} \sum_{i=1}^{N}{\lambda_i \left(Z_i+\alpha_i\right)^2}\right).
\end{align}
Using the above upper bound of the likelihood ratio, we write
\begin{align}
\mathbbm{1}_{\left(\sum\limits_{i=1}^{N}{\lambda_i (Z_i+\alpha_i)^2} \leq \gamma_0\right)} \mathcal{L}(\bm{Z}) \leq \left(\frac{\gamma_0}{N}\right)^{\frac{N}{2}} \left(\prod_{i=1}^{N}{\frac{1}{\sqrt{\lambda_i}}}\right) e^{\frac{N}{2}}.
\end{align}
Thus, we obtain the upper bound
\begin{align}\label{1}
\mathbb{E}^*[\mathbbm{1}_{\left(S_N \leq \gamma_0\right)} \mathcal{L}^2(\bm{Z})] \leq \left(\frac{\gamma_0}{N}\right)^{N} \left(\prod_{i=1}^{N}{\frac{1}{\lambda_i}}\right) e^{N}.
\end{align}
From Proposition \ref{prop1}, we have 
\begin{align}
P \geq  \prod_{i=1}^{N}{\left[1-Q_{\frac{1}{2}}\left(\alpha_i,\sqrt{\frac{\gamma_0}{N \lambda_i}}\right)\right]}.
\end{align}
Using \citep[Eq.(8)]{TSAY}, we have the  asymptotic expansion around $b=0$ of $Q_{\nu}(a,b)$, i.e.,
\begin{align}
Q_{\nu}(a,b) \underset{b \rightarrow 0}{\sim} 1-\frac{1}{\Gamma(\nu+1)} \left(\frac{b^2}{2}\right)^{\nu} \left(\frac{a^2}{2}\right)^{1-\nu}.
\end{align}
Therefore, as $\gamma_0 \rightarrow 0$, we have
\begin{align}
P \geq  \left(\frac{1}{N \pi}\right)^{\frac{N}{2}} \gamma_0^{\frac{N}{2}} \left(\prod_{i=1}^{N}{\frac{\alpha_i}{\sqrt{\lambda_i}}}\right),
\end{align}
and we can write
\begin{align}\label{2}
\frac{1}{P^2} \leq \left(N \pi\right)^N \gamma_0^{-N} \left(\prod_{i=1}^{N}{\frac{\lambda_i}{\alpha_i^{2}}}\right).
\end{align}
By combining (\ref{1}) and (\ref{2}), we obtain
\begin{align}
\underset{\gamma_0 \rightarrow 0} \limsup \frac{\mathbb{E}^*[\mathbbm{1}_{\left(S_N \leq \gamma_0\right)} \mathcal{L}^2(\bm{Z})]}{P^2} \leq \left(\prod_{i=1}^{N}{\frac{\pi e}{\alpha_i^{2}}}\right) < +\infty.
\end{align}
\end{proof}
\begin{remark}
If $\bm{\mu} = 0$, then $\bm{\alpha} = 0$, the biased PDF is Gaussian with zero mean and variance $\sigma_i^2$. In the proof, we use the following lower bound
\begin{align}
P = \mathbb{P}(\bm{X}^{T} \bm{\Sigma} \bm{X} \leq \gamma_0) \geq \frac{1}{\pi^{\frac{N}{2}}} \prod_{i=1}^{N}{\gamma\left(\frac{1}{2}, \frac{\gamma_0}{ 2 N \lambda_i}\right)},
\end{align}
where $\gamma(\cdot,\cdot)$ is the lower incomplete Gamma function \cite[Eq. (8.350.1)]{TII}. As $\gamma_0 \rightarrow 0$, we have the upper bound
\begin{align}
\frac{1}{P^2} \leq \left(\frac{N \pi}{2}\right)^N \left(\prod_{i=1}^{N}{\lambda_i}\right) \gamma_0^{-N}.
\end{align}
Our IS estimator also has the bounded relative error property in this case
\begin{align}
\underset{\gamma_0 \rightarrow 0} \limsup \frac{\mathbb{E}^*[\mathbbm{1}_{\left(S_N \leq \gamma_0\right)} \mathcal{L}^2(\bm{Z})]}{P^2} \leq \left(\frac{\pi e}{2}\right)^N < +\infty.
\end{align}
\end{remark}
\begin{remark}
Our proposed approach is still valid even if we consider the complex case for which the probability is
\begin{align}
P = \mathbb{P}(\bm{X}^{*} \bm{\Sigma} \bm{X} \leq \gamma_0),
\end{align}
where $\bm{X}$ is a complex Gaussian random vector, $\bm{X}^{*}$ is its conjugate transpose, and $\bm{\Sigma}$ is a Hermitian positive definite matrix. The complex setting appears in many applications involving wireless techniques \cite{Pablo2018}-\cite{Al-Naffouri2016}. Our estimator maintains the bounded relative error in the complex-valued case.
\end{remark}
\begin{remark}
The derived upper bound for the relative error suffers from an exponential deterioration w.r.t. $N$. Note that this means that the bound becomes useless for large $N$, but this does not mean our approach is not valid in this regime. We show in the simulations section, that even when $N$ is large, our approach still works numerically (see next section).
\end{remark}
\section{Numerical Simulations}\label{section5}
To show the accuracy and efficiency of the proposed IS scheme, we plot the CDF given in \eqref{P} using our proposed approach and we compare it to naive MC, as well as some baselines from the literature. We consider a similar setting as in \cite{Ramirez2019} where the elements of the matrices are $[\Sigma]_{i,j} = \xi^{|i-j|}$ and $[\Sigma_X]_{i,j} = \rho^{|i-j|}$ where $\xi, \rho \in (0,1)$. With this construction, we ensure that both matrices are positive definite. For the sake of simplicity, we consider the mean vector $\mu$ to be identical for all variates. All simulations were obtained using MATLAB operated on a 1.8 GHz Intel Core i5 CPU.
\subsection{Comparison with MC}
We start by showing the computational gain of our proposed approach compared to MC. To this end, we plot the CDF, given in \eqref{P}, as a function of the threshold $\gamma_{0}$ for different values of the quadrature order $N \in \{10, 20, 30\}$ when $\xi = 0.4$, $\rho = 0.8$ and $\mu_i = 1,~ \forall i$. We can clearly see that although using fewer simulation runs ($10^{4}$ less), our proposed IS estimator accurately estimates the CDF, unlike naive MC, which completely fails to estimate probabilities smaller than $10^{-7}$.
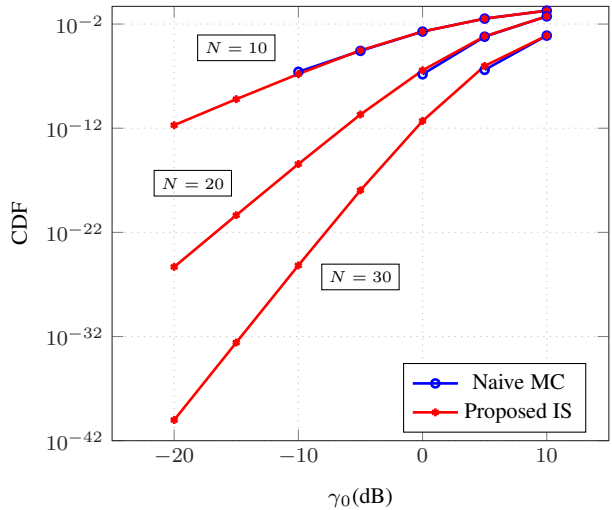
\begin{figure}[h]
\centering
\setlength\figureheight{0.35\textwidth}
\setlength\figurewidth{0.4\textwidth}
\scalefont{0.7}
\begin{tikzpicture}

\begin{semilogyaxis}[%
width=\figurewidth,
height=\figureheight,
scale only axis,
every outer x axis line/.append style={darkgray!60!black},
every x tick label/.append style={font=\color{darkgray!60!black}},
xmin=-25, 
xmax=15,
xlabel={$\gamma{}_{0}\text{(dB)}$},
xmajorgrids,
every outer y axis line/.append style={darkgray!60!black},
every y tick label/.append style={font=\color{darkgray!60!black}},
ymin=1E-42, ymax=0.5,
yminorticks=true,
ymajorgrids,
yminorgrids,
grid style={dotted},
ylabel={CDF},
legend style={at={(0.58574218750001,0.0184593784825617)},anchor=south west,draw=black,fill=white,align=left}]
\addplot [
color=blue,
solid,
line width=1.0pt,
mark size=1.5pt,
mark=o,
mark options={solid},
]
table[row sep=crcr]{%
-20	0\\
-15	0\\
-10	0.00000026\\
-5	0.00002657\\ 
0   0.0019\\
5	0.0333\\
10	0.1925\\
};
\addlegendentry{Naive MC};
\addplot [
color=red,
solid,
line width=1.0pt,
mark size=1.5pt,
mark=asterisk,
mark options={solid},
]
coordinates{
 (-20,1.9646E-12)(-15,6.1648E-10)(-10,1.6050E-07)(-5,2.9640E-05)(0,0.0019)(5,0.0333)(10,0.1925)
};
\addlegendentry{Proposed IS};
\addplot [
color=blue,
solid,
line width=1.0pt,
mark size=1.5pt,
mark=o,
mark options={solid},
forget plot
]
coordinates{
(-20,0)(-15,0)(-10,0)(-5,0)(0,1.5000E-07)(5,6.1354E-04)(10,0.0550)
};
\addplot [
color=red,
solid,
line width=1.0pt,
mark size=1.5pt,
mark=asterisk,
mark options={solid},
forget plot
]
coordinates{
(-20,4.8228E-26)(-15,4.4821E-21)(-10,3.6450E-16)(-5,2.0824E-11)(0,3.6258E-07)(5, 6.7354E-04)(10,0.0550)
};
\addplot [
color=blue,
solid,
line width=1.0pt,
mark size=1.5pt,
mark=o,
mark options={solid},
forget plot
]
coordinates{
(-20,0)(-15,0)(-10,0)(-5,0)(0,0)(5,4E-07)(10,7.8013E-04) 
};
\addplot [
color=red,
solid,
line width=1.0pt,
mark size=1.5pt,
mark=asterisk,
mark options={solid},
forget plot
]
coordinates{
(-20,9.6357E-41)(-15,2.6672E-33)(-10,6.7765E-26)(-5,1.0989E-18)(0,4.9436E-12)(5,9.5018E-07)(10,7.8013E-4) 
};
\node[black,draw,below] at (axis cs:-15,1.0e-03){{\tiny $N=10$}};
\node[black,draw,below] at (axis cs:-18.5,1.0e-16){{\tiny $N=20$}};
\node[black,draw,below] at (axis cs:-5,1.0e-25){{\tiny $N=30$}};
\end{semilogyaxis}
\end{tikzpicture}
\caption{CDF, as defined in \eqref{P}, for different values of $N$, $\xi = 0.4$, $\rho = 0.8$ and $\mu_i = 1,~ \forall i$. Number of samples $M=10^8$ and $M^*=10^4$.}
\label{fig1}
\end{figure}

We plot the required number of simulation runs for a $5\%$ accuracy requirement in Fig. \ref{fig2} where the arrows in the plot indicate the increment of the value of $N$. This plots confirms that our proposed IS scheme outperforms naive MC. For a fixed threshold $\gamma_{0}$, we notice that the number of simulation runs of IS is far less than the one needed by naive MC to achieve the same accuracy. We also note that while the number of samples of naive MC continues to increase at a high rate as the probability becomes smaller, the number of samples required by the IS estimator remains almost constant as a result of the bounded relative error property. For instance, for $N=20$, and $\gamma_0 = -5dB$, the gain in terms of number of simulation runs of our method compared to MC is around $10^{10}$.
\begin{figure}[h]
\centering
\setlength\figureheight{0.35\textwidth}
\setlength\figurewidth{0.4\textwidth}
\scalefont{0.7}
\begin{tikzpicture}

\begin{semilogyaxis}[%
width=\figurewidth,
height=\figureheight,
scale only axis,
every outer x axis line/.append style={darkgray!60!black},
every x tick label/.append style={font=\color{darkgray!60!black}},
xmin=-25, 
xmax=15,
xlabel={$\gamma{}_{0}\text{(dB)}$},
xmajorgrids,
every outer y axis line/.append style={darkgray!60!black},
every y tick label/.append style={font=\color{darkgray!60!black}},
ymin=100, ymax=1e+44,
yminorticks=true,
ymajorgrids,
yminorgrids,
grid style={dotted},
ylabel={Simulation Runs},
legend style={at={(0.58574218750001,0.724593784825617)},anchor=south west,draw=black,fill=white,align=left}]
\addplot [
color=blue,
solid,
line width=1.0pt,
mark size=1.5pt,
mark=o,
mark options={solid},
]
coordinates{
 (-20,7.6057e+14)(-15,2.5360e+12)(-10,9.4288e+09)(-5,4.9698e+07)(0,6.2041e+05)(5,3.1246e+04)(10,5.6984e+03) 
};
\addlegendentry{Naive MC};
\addplot [
color=red,
solid,
line width=1.0pt,
mark size=1.5pt,
mark=asterisk,
mark options={solid},
]
coordinates{
 (-20,3.0738e+03)(-15,3.1785e+03)(-10,3.3562e+03)(-5,4.1520e+03)(0,3.6679e+03)(5,3.0564e+03)(10,3.0918e+03) 
};
\addlegendentry{Proposed IS};
\addplot [
color=blue,
solid,
line width=1.0pt,
mark size=1.5pt,
mark=o,
mark options={solid},
forget plot
]
coordinates{
(-20,3.1862e+28)(-15,3.4284e+23)(-10,4.2158e+18)(-5,7.3793e+13)(0,4.2380e+09)(5,2.2799e+06)(10,2.7188e+04)
};
\addplot [
color=red,
solid,
line width=1.0pt,
mark size=1.5pt,
mark=asterisk,
mark options={solid},
forget plot
]
coordinates{
(-20,4.7426e+03)(-15,5.0329e+03)(-10,5.6532e+03)(-5,7.6764e+03)(0,7.8764e+03)(5,7.664e+03)(10,7.84e+03)
};
\addplot [
color=blue,
solid,
line width=1.0pt,
mark size=1.5pt,
mark=o,
mark options={solid},
forget plot
]
coordinates{
(-20,1.5947e+43)(-15,5.7612e+35)(-10,2.2676e+28)(-5,1.3984e+21)(0,3.1083e+14)(5,1.6172e+09)(10,5.8438e+05) 
};
\addplot [
color=red,
solid,
line width=1.0pt,
mark size=1.5pt,
mark=asterisk,
mark options={solid},
forget plot
]
coordinates{
(-20,6.1909e+03)(-15,6.4049e+03)(-10,7.0647e+03)(-5,8.3024e+03)(0,8.3044e+03)(5,8.1044e+03)(10,8.234e+03) 
};
\draw[thick,arrows={->}] (axis cs:-18,1.0e09) -- (axis cs:0,1.0e22);
\draw[thick,arrows={->}] (axis cs:-16,1.0e03) -- (axis cs:-19,1.0e05);
\node[black,draw,below] at (axis cs:-17,1.0e09){{\tiny $N=10, 20, 30$}};
\end{semilogyaxis}
\end{tikzpicture}%
\caption{Number of required simulation runs for $5\%$ relative error for different values of $N$, when $\xi = 0.4$, $\rho = 0.8$ and $\mu_i = 1,~ \forall i$. The arrows indicate the increment of $N$.}
\label{fig2}
\end{figure}
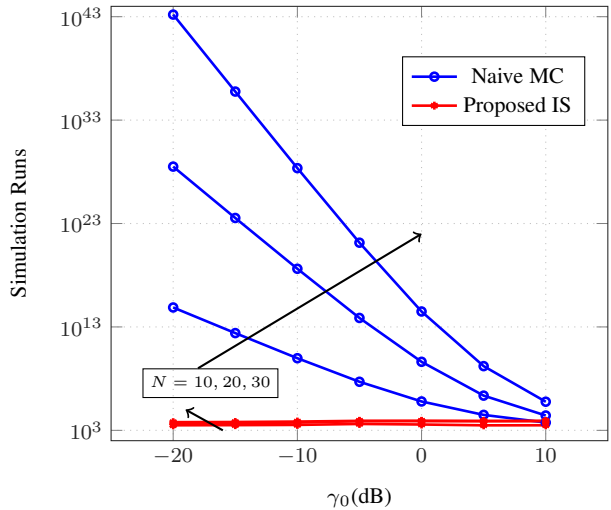
\vspace{-0.75cm}
\subsection{Comparison with other baselines}
In this section, we compare our proposed approach to the approximations presented in \cite{imhof1961}, \cite{Ramirez2019}, and the saddle-point approximation (spa) using Newton's method presented in \cite{Paolella2007}. For $\gamma_0 = 5dB$, we plot, in Figs. \ref{fig3} and \ref{fig4}, the CDF for two different sets of values of the parameters $\xi$, $\rho$, and $\bm{\mu}$ as a function of the quadrature order $N$. For the approach presented, in \cite{Ramirez2019}, we used two values of the parameter governing the accuracy of the approximation, i.e. $m=200$ and $m=500$. From Fig. \ref{fig3}, we observe that although all methods match for relatively high probabilities, they start to give inaccurate estimation of the probability of interest as soon as it becomes very small, unlike our proposed approach. In Fig. \ref{fig4}, the gap in accuracy becomes more obvious even for small probabilities, which indicates that some of these method are sensitive to the parameters of the simulation. Finally, we report, in table \ref{tab}, the average CPU time in seconds based on ten evaluations of the CDF at $\gamma_0 = 5 dB$ when $\xi=0.4$, $\rho=0.8$ and $\mu_i = 1, ~\forall i$. The results of our proposed approach are based on a number of samples equal to $M^*=10^4$. We can clearly see that our proposed approach is the fastest, followed by saddle point approximation (spa), then Imhof method. We can see that the choice of the parameter $m$ in \cite{Ramirez2019} affects also the CPU time. For all methods, the CPU times increases as the order of quadrature $N$ increases.
\section{Conclusion}
In this paper, we proposed an efficient IS estimator for the left tail of positive quadratic forms in Gaussian RVs. We showed that our estimator is endowed with the bounded relative error property, making it an appealing alternative to MC. Numerical simulations show a clear gain in terms of simulation runs, confirming the efficiency of our scheme compared to MC. Our approach also exhibits better performance than several existing methods in terms of accuracy and CPU time.
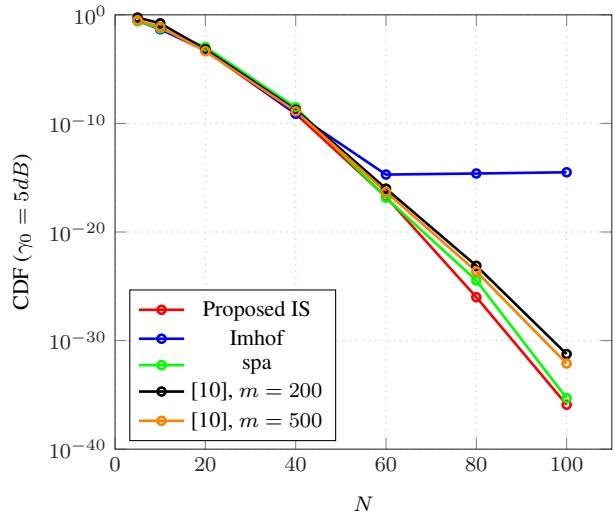
\begin{figure}[h]
\centering
\setlength\figureheight{0.35\textwidth}
\setlength\figurewidth{0.4\textwidth}
\scalefont{0.7}
\begin{tikzpicture}

\begin{semilogyaxis}[%
width=\figurewidth,
height=\figureheight,
scale only axis,
every outer x axis line/.append style={darkgray!60!black},
every x tick label/.append style={font=\color{darkgray!60!black}},
xmin=0, 
xmax=110,
xlabel={$N$},
xmajorgrids,
every outer y axis line/.append style={darkgray!60!black},
every y tick label/.append style={font=\color{darkgray!60!black}},
ymin=1E-40, ymax=1,
yminorticks=true,
ymajorgrids,
yminorgrids,
grid style={dotted},
ylabel={CDF ($\gamma_0 = 5 dB$)},
legend style={at={(0.028574218750001,0.0184593784825617)},anchor=south west,draw=black,fill=white,align=left}]
\addplot [
color=red,
solid,
line width=1.0pt,
mark size=1.5pt,
mark=o,
mark options={solid},
]
coordinates{
 (5,0.2859)(10,0.0508)(20,5.3505e-04)(40,7.7312e-10)(60,1.7191e-17)(80,9.9582e-27)(100,1.2432e-36)
};
\addlegendentry{Proposed IS};
\addplot [
color=blue,
solid,
line width=1.0pt,
mark size=1.5pt,
mark=o,
mark options={solid},
]
coordinates{
 (5,0.2882)(10,0.0472)(20,5.9554e-04)(40,7.9750e-10)(60,1.9429e-15)(80,2.4425e-15)(100,3.1641e-15)
};
\addlegendentry{Imhof};
\addplot [
color=green,
solid,
line width=1.0pt,
mark size=1.5pt,
mark=o,
mark options={solid},
]
coordinates{
 (5,0.2559)(10,0.0568)(20,0.0011)(40,3.0313e-09)(60,1.4206e-17)(80,3.4167e-25)(100,5.0322e-36)
};
\addlegendentry{spa};
\addplot [
color=black,
solid,
line width=1.0pt,
mark size=1.5pt,
mark=o,
]
coordinates{
 (5,0.562101)(10,0.165241)(20,6.83056e-04)(40,1.79933e-09)(60,9.82048e-17)(80,7.73307e-24)(100,5.74254e-32)
};
\addlegendentry{[10], $m=200$};
\addplot [
color=orange,
solid,
line width=1.0pt,
mark size=1.5pt,
mark=o,
mark options={solid},
]
coordinates{
(5,0.315604)(10,0.0784907)(20,4.4636e-04)(40,1.31875e-09)(60,4.96369e-17)(80,2.33074e-24)(100,7.70926e-33)
};
\addlegendentry{[10], $m=500$};
\end{semilogyaxis}
\end{tikzpicture}%
\caption{CDF (for $\gamma_0 = 5 dB$) when $\xi=0.4$, $\rho=0.8$ and $\mu_i = 1, ~\forall i$. Number of samples of IS $M^*=10^4$.}
\label{fig3}
\end{figure}

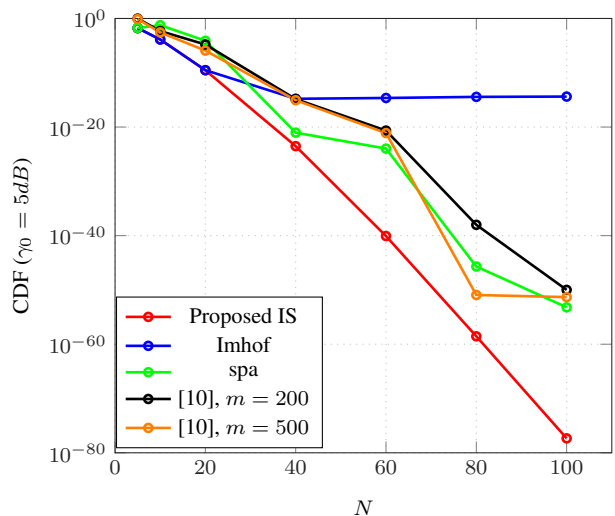
\begin{figure}[h]
\centering
\setlength\figureheight{0.35\textwidth}
\setlength\figurewidth{0.4\textwidth}
\scalefont{0.7}
\begin{tikzpicture}

\begin{semilogyaxis}[%
width=\figurewidth,
height=\figureheight,
scale only axis,
every outer x axis line/.append style={darkgray!60!black},
every x tick label/.append style={font=\color{darkgray!60!black}},
xmin=0, 
xmax=110,
xlabel={$N$},
xmajorgrids,
every outer y axis line/.append style={darkgray!60!black},
every y tick label/.append style={font=\color{darkgray!60!black}},
ymin=1E-80, ymax=1,
yminorticks=true,
ymajorgrids,
yminorgrids,
grid style={dotted},
ylabel={CDF ($\gamma_0 = 5 dB$)},
legend style={at={(0.0018574218750001,0.0104593784825617)},anchor=south west,draw=black,fill=white,align=left}]
\addplot [
color=red,
solid,
line width=1.0pt,
mark size=1.5pt,
mark=o,
mark options={solid},
]
coordinates{
 (5,0.0165)(10,1.3653e-04)(20,2.8906e-10)(40,3.0297e-24)(60,8.3469e-41)(80,2.8498e-59)(100,4.5781e-78)
};
\addlegendentry{Proposed IS};
\addplot [
color=blue,
solid,
line width=1.0pt,
mark size=1.5pt,
mark=o,
mark options={solid},
]
coordinates{
 (5,0.0174)(10,1.3375e-04)(20,3.0066e-10)(40,1.6098e-15)(60,2.2760e-15)(80,3.6082e-15)(100,4.1078e-15)
};
\addlegendentry{Imhof};
\addplot [
color=green,
solid,
line width=1.0pt,
mark size=1.5pt,
mark=o,
mark options={solid},
]
coordinates{
 (5,0.0179)(10,0.0571)(20,7.7984e-05)(40,9.2916e-22)(60,1.0742e-24)(80,2.0334e-46)(100,6.1835e-54)
};
\addlegendentry{spa};
\addplot [
color=black,
solid,
line width=1.0pt,
mark size=1.5pt,
mark=o,
]
coordinates{
 (5,0.9556)(10,4.83e-03)(20,1.59575e-05)(40,1.42963e-15)(60,2.23424e-21)(80,9.78779e-39)(100,1.02277e-50)
};
\addlegendentry{[10], $m=200$};
\addplot [
color=orange,
solid,
line width=1.0pt,
mark size=1.5pt,
mark=o,
mark options={solid},
]
coordinates{
(5,0.77404)(10,2.73e-03)(20,1.20944e-06)(40,8.70858e-16)(60,7.82196e-22)(80,1.16934e-51)(100,4.69946e-52)
};
\addlegendentry{[10], $m=500$};
\end{semilogyaxis}
\end{tikzpicture}%
\caption{CDF (for $\gamma_0 = 5 dB$) when $\xi=0.1$, $\rho=0.5$ and $\mu_i = 2, ~\forall i$. Number of samples of IS $M^*=10^4$.}
\label{fig4}
\end{figure}

\begin{table}[H]
\centering
\begin{tabular}{ c c c c c c c }
\hline 
$N$                 & 10    & 20    & 40    & 60    & 80    & 100   \\ \hline \hline 
IS ($M^*=10^4$)    & 0.053 & 0.055 & 0.063 & 0.068 & 0.072 & 0.078 \\  
Imhof             & 0.323 & 0.357 & 0.361 & 0.432 & 0.532 & 0.632 \\  
spa               & 0.073 & 0.078 & 0.083 & 0.086 & 0.092 & 0.096 \\  
{[}10{]}, $m=200$ & 0.282 & 0.306 & 0.359 & 0.404 & 0.435 & 0.486 \\  
{[}10{]}, $m=500$ & 2.32  & 2.34  & 2.39  & 2.5   & 2.72  & 2.82  \\ \hline
\end{tabular}
\caption{Average CPU time (s) for a single computation of the CDF at $\gamma_0 = 5 dB$ when $\xi=0.4$, $\rho=0.8$ and $\mu_i = 1, ~\forall i$.}
\label{tab}
\end{table}

\vspace{-1.2cm}
\nocite{*}
\bibliographystyle{IEEEtran}
\bibliography{reference}
\end{document}